\nonstopmode

\documentclass[11pt]{article}
\usepackage{natbib}
\usepackage{url}

\usepackage{dcolumn}

\usepackage{amsthm}
\usepackage{amsmath}
\usepackage{amsbsy}
\usepackage{amsfonts}
\usepackage{bbold}

\usepackage[normalem]{ulem}
\useunder{\uline}{\ul}{}

\usepackage{graphicx}
\usepackage{color}
\DeclareMathSymbol{\real}{\mathord}{AMSb}{"52}
\DeclareMathSymbol{\natural}{\mathord}{AMSb}{"4E}
\DeclareMathSymbol{\prob}{\mathord}{AMSb}{"50}
\DeclareMathSymbol{\blackbox}{\mathord}{AMSa}{"04}

\usepackage{accents}
\newlength{\dhatheight}

\usepackage[pdftex]{hyperref}

\usepackage{setspace}
\newtheorem{proposition}{Proposition}

\newtheorem{remark}{Remark}

\newenvironment{definition}{\vspace{.4cm} \noindent {\bf Definition:}}{\vspace{.4cm}}

\newif\ifanonymous
\anonymousfalse

\begin{document}

\ifanonymous
\title{From Means to Medians: Optimal Benchmark Design}
\author{}
\date{}
\else
\title{From Means to Medians: Optimal Benchmark Design\thanks{The author declares that he has no known competing financial interests or personal relationships that could have appeared to influence the work reported in this paper.}}

\author{\'Angel Hernando-Veciana\thanks{Universidad Carlos III de Madrid, calle Madrid 126, Getafe (Madrid) 28903, Spain. Email: angel.hernando@uc3m.es. I would like to thank Marcelo Bagnulo, William Fuchs, Efthymios Smyrniotis and seminar participants for their comments and suggestions. This research is part of the project I+D+i TED2021-131844B-I00, funded by MCIN/ AEI/10.13039/501100011033 and the European Union NextGeneration EU/PRTR.}\\Universidad Carlos III de Madrid}

\date{\today}
\fi

\maketitle

\begin{abstract}
Manipulation of price benchmarks has cost hundreds of millions of dollars, making manipulation-resistant design a first-order question. A puzzling empirical pattern offers a starting point: benchmarks in traditional finance are typically means, those in decentralised finance medians. I rationalise this divide with a variable cost rising in each price distortion's size and a fixed cost per manipulated price. The optimal benchmark is a mean when fixed costs are negligible and variable costs convex, a median when variable costs are negligible, and otherwise a trimmed mean, with weights reflecting cost heterogeneity.

\textbf{Key Words:} Financial benchmarks, oracles, decentralised finance.

\medskip
\noindent\textbf{JEL Classification:} G14, D47, D82, D86.

\end{abstract}

\section{Introduction}

A (price) \emph{benchmark} provides a summary measure of the prices prevailing in a given market at a point in time. Benchmarks are computed for most major financial markets; prominent examples include the ``4pm Fix,'' which reports daily exchange rates from foreign-exchange transactions sampled around 4pm London time, see \cite{ber16}, and the original LIBOR and EURIBOR, derived from opinion polls of panel banks, see \cite{hertro20}.\footnote{Opinion-based benchmarks are also used in some smart contracts; see \cite{williams2019decentralized}. A well-known example is Augur; see \cite{peterson2015augur}.}

By allowing contracts to be indexed to prevailing market prices, benchmarks underpin a wide range of financial instruments, including derivatives, floating-rate loans, and portfolio valuations. The notional value of such contracts can be large enough that one contracting party may have an incentive to shift the benchmark in its favour. There is ample evidence that such manipulation has occurred; see \cite{schros13} for prominent media accounts and \cite{eva18} for an academic investigation.

More recently, the advent of decentralised finance (DeFi) has significantly amplified these concerns. The combination of large contract values, automated execution, and limited legal enforceability makes manipulation of price oracles potentially profitable and difficult to deter. In this environment, smart contracts rely critically on external inputs in the form of financial benchmarks, commonly referred to as price oracles.\footnote{\cite{angeris2020improved} studies in detail examples of price oracles computed using transaction prices, including the Uniswap oracle developed in \cite{adams2020uniswap}. The term ``price oracle'' is sometimes used narrowly to refer to protocols that import external information into a blockchain, and more broadly to any protocol broadcasting benchmark values within the blockchain ecosystem. Under the latter interpretation, one can distinguish between on-chain and off-chain oracles; see \cite{weretal22}.} The most widely used platform for these applications is the Ethereum blockchain, which processes millions of transactions per day and supports very large gross on-chain transfer flows.

Consider, for instance, one of the most basic financial contracts: lending. Due to the lack of enforcing courts, lending is always collateralised and an oracle is used to determine the value of the collateral. This creates two distinct manipulation incentives. First, a borrower may inflate the oracle to borrow more than the collateral is worth. In a well documented case, the Mango Markets exploit, a trader used \$10 million to manipulate the platform's price oracle, inflating the value of a token used as collateral, to borrow and withdraw over \$116 million and then defaulted leading to significant platform losses.\footnote{See \url{https://infotrend.com/mango-markets-madness-a-case-study-on-the-mango-markets-exploit/} and \url{https://www.sec.gov/newsroom/press-releases/2023-13} for a description.} Second, since lending contracts have automatic liquidation clauses that can be triggered by any external party when the collateral value falls below a threshold, it may equally be profitable to deflate the oracle to activate liquidation against another party, see \cite{weretal22} for documented examples.

More broadly, the development of DeFi hinges on the availability of reliable oracles that can be safely incorporated into smart contracts. As the Financial Conduct Authority, a regulatory body of the UK Government, puts it:\footnote{See \cite{humphry2024review}.}
\begin{quote}
``Oracles and bridges are two of the most significant vulnerabilities that exist in the DeFi world today,'' and ``any inaccuracies in the information that oracles provide can instantly impact downstream applications that rely on the oracle.''
\end{quote}
Manipulation is therefore not a theoretical curiosity but a first-order concern in benchmark design, with documented losses running into hundreds of millions of dollars.\footnote{\cite{chainalysis2023oracle} report that DeFi protocols lost \$403.2 million across 41 oracle manipulation attacks in 2022 alone.} This paper asks directly: what benchmark formula best resists manipulation?

My model answers this question by establishing a general correspondence between the cost structure of manipulation and the optimal benchmark. This framework yields, as a natural corollary, an explanation for why the mean is more prevalent in traditional finance whereas the median dominates in DeFi.\footnote{For instance, the 4pm Fix (see \url{https://www.lseg.com/en/ftse-russell}), the EURIBOR (see \url{https://www.emmi-benchmarks.eu/benchmarks/euribor/methodology/}) and the USD LIBOR (see \url{https://www.ice.com/publicdocs/USD_LIBOR_Methodology.pdf}) use a (weighted) mean, see \cite{baldauf2022principal}, whereas popular oracles like MakerDAO or Tellor use the median, see \cite{eskandari2021sok}.}

Unlike \cite{baldauf2022principal}, whose analysis focuses exclusively on convex variable manipulation costs, costs that increase with the magnitude of the price distortion at an increasing rate, I recognise two further features of manipulation costs that are relevant in practice. First, manipulation may also entail fixed costs that are independent of the size of the distortion but increase with the number of prices manipulated. This is particularly relevant when prices are determined on decentralised exchanges on a blockchain, which is the typical setting in DeFi. In this environment, manipulating each price requires submitting at least one transaction to the blockchain, incurring a cost that is independent of the extent of the price manipulation.\footnote{While the most visible component is the so-called \emph{gas} cost (see \cite{perliv19}), the economically more significant costs in many manipulative attacks are bribes to miners. These consist of direct payments conditional on enforcing a particular transaction ordering within a block. In current blockchain systems, such payments are largely mediated through MEV auctions; see \cite{daia20flash} and \cite{wahrstaetter2023timetobribe}.\label{foogas}} Remarkably, variable manipulation costs may be negligible in DeFi due to the availability of tools such as flash loans, as argued by \cite{oos21} and \cite{arora2024secplf}.\footnote{That said, \cite{adams2020uniswap} argues that the presence of variable manipulation costs may depend on the type of prices used in the benchmark, and that such costs are more likely to arise when end-of-block prices are employed.} By contrast, fixed manipulation costs may be very large, often far exceeding standard transaction fees, as extensively documented in the MEV literature; see \cite{wahrstaetter2023timetobribe} and \cite{eskandari2021sok}. Second, convexity of the variable costs is not always well-founded: as I show in Appendix~\ref{appCPMM}, when variable costs stem from the trading fee charged by a constant-product automated market maker, they turn out to be concave in the magnitude of manipulation.

My first main result (Proposition~\ref{profixvar}) establishes that the optimal benchmark is the median when variable costs are negligible, and the mean when fixed costs are negligible \emph{and} variable costs are convex. This result is intuitive. When variable costs are negligible, the main concern arises from large distortions in a small number of prices, which are better mitigated by the median than by the mean. When fixed costs are negligible and variable costs are convex, the cheapest manipulation strategy is instead to spread the bias uniformly across all prices, concentrating it on a few would entail disproportionately high marginal costs. The mean is maximally robust to such uniform manipulation, whereas the median is highly sensitive to manipulations that affect a large fraction of prices.

Establishing this result in full generality is technically challenging, as the problem formally corresponds to an optimal contracting problem with multitask moral hazard. In this setting, the manipulator plays the role of the agent and chooses how much to distort each of the market prices that enter the benchmark. The analysis becomes even more involved in the presence of fixed costs or concave variable costs, since the manipulator's objective function is non-convex, even under linear contracts. As a result, the manipulator's optimal behavior cannot be fully characterized using first-order conditions alone.

My approach takes a different perspective. Consider, first, the case in which fixed costs are absent and the benchmark is the mean. With convex variable manipulation costs, the optimal strategy for the manipulator consists of uniformly increasing all market prices by the same amount. Such a uniform manipulation would increase by the same amount any benchmark satisfying basic and reasonable properties. For example, \emph{symmetric benchmarks}, those that coincide with the point of symmetry when the underlying price distribution is symmetric, would also increase by the same amount if the original distribution were symmetric. However, other symmetric benchmarks may admit alternative manipulation strategies that achieve the same effect at a lower cost. In this sense, the mean is optimal, as it maximizes the cost of manipulation among all symmetric benchmarks.

The case without variable costs is analogous, though slightly more complex. In this case, the optimal manipulation strategy of the median involves increasing a sufficient mass of prices so that only half of the total mass remains below the target value for the benchmark. With an appropriate adjustment, this minimal mass can be shifted in a way that preserves symmetry, ensuring that if the original price distribution is symmetric, the manipulated distribution remains symmetric as well. As in the previous case, this implies that the median maximizes the cost of manipulation among all symmetric benchmarks, and is therefore optimal in this sense.

Extending the analysis beyond these extreme cases is challenging. However, in the specific scenario where the initial distribution of prices is concentrated at a single point, one can generally characterize benchmarks that induce optimal manipulations that preserve symmetry, thereby ensuring optimality among symmetric benchmarks. For the case of convex variable costs, this turns out to be trimmed means\footnote{A trimmed mean is calculated by excluding all values below and above a given quantile from the computation of the mean.}, where the extent of trimming increases as the fixed costs grow larger relative to the variable costs. For the case of concave manipulation costs, the result is qualitatively similar: the optimal benchmark is again a trimmed mean, with the degree of trimming increasing in the size of the fixed costs. The key difference from the convex case is that some trimming is optimal even in the absence of fixed costs. A comparative static concerning market depth further sharpens the picture: when the variable cost arises from trading against a constant-product market maker, the optimal degree of trimming decreases with the liquidity of the pool, so prices from deeper pools, where manipulation is more costly at the margin, are aggregated closer to the mean, while shallower pools call for more trimming and, in the limit, the median.

The results above assume that all prices in the sample are drawn from the same population and face the same manipulation costs. In practice, however, benchmarks are often computed from samples that span different market conditions, for instance transactions recorded at different times of day when liquidity and trading costs vary. Such variation naturally translates into heterogeneity in manipulation costs: prices sampled during illiquid periods are cheaper to distort, while the cost of securing transaction ordering in a blockchain can fluctuate sharply across time. A natural question is whether a benchmark designer can exploit this heterogeneity to further deter manipulation.

My analysis shows that the answer is yes, and the optimal adjustment is intuitive. When manipulation costs are purely variable and convex, the benchmark should downweight prices that are cheaper to manipulate at the margin, leading to a weighted mean with weights proportional to marginal manipulation costs. When costs are purely fixed, what matters instead is how costly it is to include a given price in a manipulation campaign at all, leading to a weighted median in the spirit of \cite{edgeworth1888xxii},\footnote{\cite{edgeworth1888xxii} proposed combining observations from populations with different dispersions by assigning greater weight to more precisely measured sources. My weighted median is in this spirit: prices from subpopulations with higher fixed manipulation costs are harder to include in a manipulation campaign and therefore receive greater weight. The key difference is that my weights follow from a formal manipulation-resistance criterion rather than from statistical efficiency.} with weights proportional to the fixed costs of manipulation. I also show how these insights generalize when both cost components are present, yielding a trimmed mean whose weights and trimming parameter reflect both variable and fixed costs.

The most closely related paper is that of \cite{baldauf2022principal}, to which I have already referred. The main differences are twofold. First, it assumes that manipulation entails no fixed costs that grow with the number of prices manipulated. Second, it restricts attention to convex variable costs, whereas I show that concavity is the empirically relevant case when variable costs stem from the trading fee of a constant-product automated market maker. Extending the analysis along either dimension is technically challenging: once fixed costs are present the manipulator's problem is no longer convex, ruling out the first-order approach standard in the optimal contract design literature, see \cite{bolton2004contract}.

Also closely related are \cite{duffie2021robust} and \cite{frei2021optimal}, who study benchmark design when the manipulator's only margin is the volume of distorting trades. In \cite{duffie2021robust} the manipulator chooses the volume of its orders, which are all executed at a common distorted price, at a cost linear in that volume; in \cite{frei2021optimal} it allocates an exogenously fixed total volume of distorting trades across periods. In both, manipulation is parameterised by traded volume, whereas in my model the manipulator chooses the magnitude of each individual price distortion subject to an explicit cost of manipulation comprising variable and fixed components.

More broadly, my work connects to the economics literature on decision-making from manipulable data, which stresses that when allocations depend on data, the agents who generate that data distort it in anticipation of how it will be used; see, for instance, \cite{frankel2022improving}. A benchmark is precisely such a data-dependent decision rule, and my analysis asks how its design shapes the incentives to manipulate the underlying prices.

My work is also related to the literature on robust statistics, see \cite{huber} for a comprehensive review. This literature studies the statistic that minimizes deviations from the true value subject to a perturbation chosen to maximize deviations. This problem, like mine, is formulated as a minimax, but the difference is that the assumptions that define the set of feasible deviations are reasonable for exogenous outliers but not for the manipulation problems that arise in financial markets. In particular, these deviations are not consistent with the existence of a fixed cost of manipulation, one of the main elements of my analysis. A second difference is that robust statistics weighs robustness against \emph{statistical efficiency}: the precision of the estimator when the model holds and no contamination occurs. I, instead, focus exclusively on robustness to manipulation, taking the point of symmetry as the target and asking which benchmark is hardest to manipulate. These two concerns are complementary rather than competing. Because the benchmarks that emerge as manipulation-optimal (means, trimmed means, and medians, possibly weighted) are exactly the L-estimators\footnote{An L-estimator is a linear combination of order statistics; equivalently, a location functional of the form $\int_0^1 J(u)\,G^{-1}(u)\,du$ for a weight function $J$ on the quantiles. The mean, median, and trimmed means are the leading examples, corresponding to uniform weights, all weight at the median, and uniform weights on a central quantile range, respectively.} for which the efficiency--robustness trade-off is classically understood, my manipulation-resistance criterion can be combined with the asymptotic- or minimax-efficiency criteria of robust statistics to discipline benchmark design along both dimensions at once.

My work is also related to the fast growing literature on DeFi and in particular to the analysis of oracles and their vulnerabilities, see \cite{weretal22}. This literature, however, does not analyse the performance of the possible rules used to aggregate prices but instead focuses on an empirical investigation of the performance of the different oracles used in reality, see \cite{liu2021first}, or classify the different oracles used in practice with respect to different technical aspects of their protocols, see \cite{pasdar2021blockchain} for a comprehensive discussion. This literature is also interested in the implementation of price oracles not based on transactions but on honest reporting by trusted parties, see for instance \cite{buterin2014schellingcoin}.

Finally, my analysis also complements work on securing specific DeFi applications, such as \cite{arora2024secplf}, who deter over-borrowing in protocols for loanable funds by wrapping the oracle in an asymmetric rule that caps upward moves in collateral value, at the cost of distorting its valuation. A benchmark that is already optimally manipulation-resistant reduces the distortion such a wrapper must impose, and analogous wrappers could address strategic liquidations or other directional attacks.

\ifanonymous
All proofs are in Appendix~\ref{appProofs}.
\else
The remainder of the paper is organised as follows. Section~\ref{sec:Model} presents the model. Section~\ref{secOpt} derives the optimal benchmark in the baseline case of a homogeneous market, characterising the mean, median, and trimmed mean as optimal depending on the relative importance of fixed and variable manipulation costs. Section~\ref{secHet} extends the analysis to heterogeneous markets and establishes the optimality of weighted means, weighted medians, and weighted trimmed means. Section~\ref{secConclusion} concludes. All proofs are in Appendix~\ref{appProofs}. Appendix~\ref{appCPMM} derives the manipulation cost function implied by a constant-product market maker and verifies that it satisfies the concavity and regularity conditions used in the analysis.
\fi

\section{The Model}

\label{sec:Model}

I assume a population of {\it unmanipulated prices} described by a probability distribution $F$. For some results I assume that the distribution is a degenerate distribution at a price; I refer to this case as the case of no dispersion. For the rest of the paper, I assume that $F$ has a symmetric and strictly single-peaked\footnote{A function $u:\mathbb{R}\to\mathbb{R}$ is \emph{strictly single-peaked} if there exists a unique point $x^*\in\mathbb{R}$ such that
\begin{itemize}
\item for all $x<y\le x^*$, $u(x)<u(y)$,
\item for all $x^*\le y<z$, $u(y)>u(z)$.
\end{itemize}} density $f$.

The restriction to symmetric distributions is both natural and well-motivated. For asymmetric distributions, standard location functionals, such as the mean, median, and trimmed mean, generally disagree, and there is no canonical choice: the very meaning of ``location'' depends on the purpose of the analysis rather than being an intrinsic property of the distribution. Restricting to symmetric distributions thus isolates the comparison of benchmarks to their manipulation-resistance properties, without conflating it with the separate question of which functional best represents location under asymmetry.

A strategic agent, {\it the manipulator}, can manipulate these prices at a cost. In particular, the manipulator's strategy can be described with a conditional distribution function $G^c(\cdot|p)$, for any $p\in \real$, that describes how the prices of all the markets with unmanipulated price $p$ are shifted into a new distribution of prices. In the case of no manipulation, $G^c(\cdot|p)$ is simply a degenerate distribution at $p$. The manipulator's strategy gives rise to the {\it distribution of manipulated prices}:
\[G(p)\equiv \int G^c(p|\tilde p) dF(\tilde p).\]

A {\it benchmark} is a locational functional applied to a distribution of prices. In my analysis, I either require very little structure on the set of rules that may be used to compute the benchmark or restrict attention to particular families of rules that have been used in practice. In particular, I call {\it symmetric benchmarks} those that are equal to the point of symmetry in the case of a symmetric distribution of prices. This is a minimal restriction imposed on location functionals, see for instance \cite{jaeckel1971}. Besides, it does not impose any structure on distributions that are not symmetric. I define the $\tau$-trimmed mean, for $\tau \in [0, 0.5]$, as the average calculated over the modified distribution obtained by truncating the original distribution at the $\tau$ and $1-\tau$ quantiles. Clearly, $\tau=0$ and $\tau=1/2$ render the $\tau$-trimmed mean equal to the standard mean and the median, respectively. The $\tau$-trimmed means are symmetric benchmarks.\footnote{The trimmed means form a canonical family in robust statistics, where the trimming fraction $\tau$ parameterises a trade-off between statistical efficiency, maximised by the mean ($\tau=0$) under light-tailed price noise, and robustness, which increases with $\tau$. The same family parameterises the trade-off relevant here, between resistance to manipulation under variable versus fixed costs.}

As a realistic example, one can think of the distribution of prices as the marginal prices offered by a set of AMMs at different moments in time. These marginal prices can be manipulated by submitting trade orders to the AMM whose only purpose is to upset the computed benchmark. Thus, the manipulated prices correspond to the marginal prices displayed by the AMM's and used in the computation of the benchmark and the unmanipulated prices refer to the counterfactual of what could have been the marginal prices if there were no manipulation. Alternatively, one could assume that the benchmark is computed out of a sample of transaction prices and manipulation is done by submitting additional trades.

Manipulating a price is costly. In particular, the cost of a manipulating strategy $G^c(\cdot|p)$ is equal to:
\[\int \int_{{\tilde p\not =p}} \left( c(\tilde p-p) +k \right) dG^c(\tilde p|p) dF(p),\]
where $c$ is a twice differentiable and strictly increasing function (except when $c(\cdot)=0$ is assumed), equal to zero at zero and symmetric around zero,\footnote{The symmetry of $c$ around zero means that the variable cost of a price distortion depends only on its magnitude, not its direction. This is natural when manipulation requires trading an asset: buying or selling the same quantity incurs the same transaction cost, and the resulting price impact is symmetric in absolute terms.} and $k\geq 0$. Thus, $c$ is the {\it variable cost} of manipulation and $k$ the {\it fixed cost} of manipulation.\footnote{The fixed cost $k$ is incurred once per manipulated price, irrespective of the size of the distortion. This per-price structure is natural in the DeFi context. Benchmarks typically aggregate prices sampled at different times or across different venues, so each manipulated observation requires securing a separate, independently arranged transaction ordering, often in a different block. Even when multiple manipulated transactions fall within the same block, the cost of securing their combined ordering scales with the number of such transactions.} The cost of manipulation is explained in that manipulating a price requires trading the asset. The variable cost arises because greater manipulation requires larger orders, which typically implies higher costs.\footnote{\cite{duffie2021robust} assumes variable manipulation costs that are linear in the traded volume, and \cite{baldauf2022principal} provides examples in which they are increasing and convex. I show in Appendix~\ref{appCPMM} that for the particular problem of variable costs originated by the trading fee charged by an automated market maker (AMM), the variable costs turn out to be concave if the AMM uses the constant product market maker.} Fixed costs, instead, stem from the need to secure transaction inclusion or priority independently of the size of the trade. For instance, in the context of blockchains, manipulating a price immediately before it is fed into a benchmark requires enforcing a particular ordering of transactions within a block. As explained in the introduction, this typically entails paying a potentially large fee to the block builder.

When the distribution of unmanipulated prices is symmetric, any symmetric benchmark coincides with the axis of symmetry. The bias induced by manipulation is therefore given by the difference between the value of the benchmark under the manipulated distribution and the axis of symmetry of the unmanipulated distribution of prices. The focus on symmetric benchmarks allows me to abstract from normative questions about how benchmarks should respond to price asymmetries and ensures that differences in manipulability arise solely from strategic distortions.

\section{Optimal Benchmarks}
\label{secOpt}

The criterion I use to evaluate a benchmark is the minimum cost of manipulation
necessary to implement a given bias $\Delta$: the larger this minimum cost, the
better the benchmark. This criterion is consistent with the collateralised lending
example discussed in the introduction, where the manipulator minimizes the cost of
triggering the liquidation clause while the benchmark designer seeks to maximize it.
This perspective is closely related to the Grossman--Hart approach to moral hazard
\citep{grohar83}, in which the principal first characterizes the minimum payment
required to induce each effort level, then selects the effort level that maximizes
profits. Analogously, each benchmark generates an endogenous cost function that
determines, for any target bias $\Delta$, the minimum cost a manipulator must incur
to implement it.

\begin{definition}$\,$
\begin{itemize}
\item A benchmark is optimal at $\Delta$ if it is the symmetric benchmark that maximizes the minimum cost of inducing bias $\Delta$. A benchmark is optimal if it is optimal at every $\Delta$.
\item A benchmark dominates another benchmark at $\Delta$ if the former has a greater minimum cost of implementing the bias $\Delta$ than the latter. A benchmark dominates another benchmark if it dominates the latter for every $\Delta$.
\end{itemize}
\end{definition}

By the symmetry of $c$ around zero and of $F$ around its centre, the minimum cost of implementing bias $-\Delta$ equals that of implementing $+\Delta$. It therefore suffices to characterise optimal benchmarks for $\Delta>0$, which is what I do throughout.

The analysis is based on the following family of manipulative strategies.

\begin{definition}
A manipulative strategy is {\it symmetric} if it preserves the symmetry of the distribution of prices, i.e. it induces a distribution of manipulated prices that is symmetric.
\end{definition}

The following natural result will be used throughout the analysis.

\begin{proposition}
A benchmark $B$ is optimal at $\Delta$ if there exists a manipulation strategy that (i) is cost-minimising for $B$ at $\Delta$, i.e.\ achieves bias $\Delta$ for $B$ at the lowest possible cost, and (ii) is symmetric, i.e.\ induces a symmetric distribution of manipulated prices.
\label{prosym}
\end{proposition}

A manipulator employing such strategy implements the bias $\Delta$ in any symmetric benchmark, but it is only cost minimizing in the proposed benchmark, and thus any other symmetric benchmark can be manipulated at a potentially lower cost.

Note that the proposition does not rule out optimal manipulation strategies that are asymmetric. Indeed, the most natural manipulation may not be symmetric. I discuss this point in detail, with a concrete example, after Proposition~\ref{profixvar} below.

My first main result covers the polar cases in which either fixed costs or variable costs are absent, where the case of no fixed costs requires the additional assumption of convex variable costs.

\begin{proposition}
The mean is an optimal benchmark in the absence of fixed costs, i.e.\ $k=0$, provided variable costs are convex. The median is an optimal benchmark in the absence of variable costs, i.e.\ $c(\cdot)=0$.
\label{profixvar}
\end{proposition}

A few remarks on Proposition~\ref{profixvar} are in order. The proof identifies,
in each case, a symmetric cost-minimizing manipulation strategy, which by
Proposition~\ref{prosym} suffices to establish optimality. However, the result
does not claim that this strategy is the only cost-minimizing one, nor that
the optimal manipulation must be symmetric. To see this concretely, consider
the median in the absence of variable costs. As the proof makes explicit, any
manipulation strategy that transfers exactly $F(P+\Delta)-1/2$ mass from strictly
below $P+\Delta$ to points weakly above it achieves the same minimum cost,
regardless of how the destinations are chosen; the proof simply selects one
such strategy that happens to preserve symmetry. Now suppose a small variable
cost is added. Because the variable cost is strictly increasing in the distance
moved, it becomes strictly cheaper to relocate each source as little as possible.
The optimal strategy therefore concentrates all manipulated mass at $P+\Delta$
itself, producing a final distribution that is highly asymmetric.

One might therefore worry that the optimality result is fragile: if the symmetric
manipulation strategy is not the unique optimum, perhaps a small perturbation of
the environment would cause the median to lose its optimality. This concern does
not arise. The symmetric strategy from the unperturbed problem remains feasible
under the perturbation, and its cost increases by at most a term proportional to
the size of the perturbation. Since no symmetric benchmark can have a higher
minimum manipulation cost than this, the median remains approximately optimal,
with the approximation error vanishing as the perturbation vanishes.

The general case in which both fixed and variable manipulation costs are present is substantially more complex. Intuitively, it corresponds to a multitask moral hazard problem in which the agent faces nonconvex preferences, even under linear contracts. Nevertheless, there are two particular cases in which Proposition~\ref{profixvar} can be extended and Proposition~\ref{prosym} continues to apply.

Let
\[
\overline c(q)\equiv \frac{k+c(q)}{|q|}
\]
denote the average cost of manipulating a single price by an amount $q$, and let $\delta_{\min}$ denote the minimum of $\overline c(q)$ over $q\in[0,\infty)$. Strict convexity of $c$ implies that this minimum is attained at a unique $q$ unless $\overline c$ is everywhere decreasing, in which case I set $\delta_{\min}=\infty$. Similarly, strict concavity of $c$  implies that $\delta_{\min}=\infty.$\footnote{Note that $\bar c'(q)=\frac{c'(q) - \bar c(q)}{q}$, so at any critical point
$q_0$ where $\bar c'(q_0)=0$ one has $c'(q_0)=\bar c(q_0)$ and thus
$\bar c''(q_0)=\frac{c''(q_0)}{q_0}$, which is strictly positive under convexity and
strictly negative under concavity. Hence, under convexity every critical point is a strict
local minimum, whereas under concavity every critical point is a strict local maximum. In
either case, two critical points would require a critical point of the opposite kind in
between, so there is at most one critical point. Under concavity, $\bar c'(q)=(c'(q)q - k - c(q))/q^2$. Strict concavity with $c(0)=0$
gives $c(q)>qc'(q)$ for all $q>0$, so $c'(q)q-c(q)<0$; together with $k\geq 0$ this
yields $\bar c'(q)<0$ for all $q>0$. Hence $\bar c$ is strictly decreasing and
$\delta_{\min}=\infty$. Under convexity, if no critical point exists, an analogous
argument shows that $\bar c$ must be strictly decreasing, and the convention $\delta_{\min}=\infty$ applies again.}

\begin{proposition}
The mean is an optimal benchmark if variable costs are convex and either $\delta_{\min}\leq \Delta$ or $\delta_{\min}=2\Delta$.
\label{prodP}
\end{proposition}

In both cases covered by the proposition, the mean admits an optimal manipulation strategy that preserves the symmetry of the price distribution. This property generally fails when either $\delta_{\min}\in(\Delta,2\Delta)$ or $\delta_{\min}>2\Delta$. The intuition is easiest to see in the case in which the unmanipulated price distribution exhibits no dispersion. Strict convexity of the variable cost function implies that the cost-minimizing manipulation of the mean consists of increasing all manipulated prices by the same amount. If $\delta_{\min}\geq \Delta$, the optimal strategy is therefore to raise a mass $\Delta/\delta_{\min}$ of prices by $\delta_{\min}$. Hence, when $\delta_{\min}\in(\Delta,2\Delta)$ the measure of manipulated prices satisfies $\mu\in(1/2,1)$, whereas when $\delta_{\min}>2\Delta$ one has $\mu<1/2$. In both cases, the resulting distribution of manipulated prices has two atoms, one at $P$ and one at $P+\delta_{\min}$, and is therefore not symmetric. The next proposition characterizes the optimal benchmark in the latter case.\footnote{Extending this analysis to more general price distributions is challenging. Even when restricting attention to commonly used benchmarks such as $\tau$-trimmed or $\tau$-Winsorized means, characterizing the optimal manipulation strategy is difficult, as the manipulator's problem becomes a nonseparable and nonconvex functional optimization problem.}

\begin{proposition}
\label{prod2P}
Suppose the distribution of unmanipulated prices has no dispersion. Then the
optimal symmetric benchmark is the $\tau^*$-trimmed mean where
\[
\tau^*\equiv \frac{1}{2}\left(1-\frac{c'(2\Delta)}{\overline c(2\Delta)}\right)\in(0,1/2),
\]
if either:
\begin{itemize}
\item $c$ is strictly convex, $k>0$, and $\delta_{\min}>2\Delta$;
\item $c$ is strictly concave, $k\geq 0$, and $\overline c'(q)/c'(q)$ is strictly increasing.
\end{itemize}
\end{proposition}

When fixed manipulation costs are large relative to the target price increase
($\delta_{\min}>2\Delta$) and variable costs are convex, the optimal manipulation of the
mean concentrates on less than half the prices, breaking the symmetry that makes the mean
optimal when $\delta_{\min}=2\Delta$. A $\tau$-trimmed mean counteracts this incentive by
discarding extreme prices, thereby forcing the manipulator to trade off the gains from
concentrating manipulation against the higher average cost of doing so. Convexity of the
variable cost implies that any optimal manipulation raises all affected prices by the
same amount, and the optimal trimming parameter $\tau^*$ is precisely such that this
amount equals $2\Delta$. At this value, the least-cost manipulation strategy shifts
exactly half of the prices by $2\Delta$, restoring symmetry around $\Delta$ and allowing
Proposition~\ref{prosym} to apply. This logic fails when $\delta_{\min}\in(\Delta,2\Delta)$:
in that case, the optimal manipulation for any trimmed mean raises prices by less than
$2\Delta$, so the mass of manipulated prices is strictly above one half and symmetry
cannot be achieved. Nevertheless, even in this intermediate region the mean dominates all
other trimmed means, consistent with Propositions~\ref{prodP} and~\ref{prod2P} (note that
$\delta_{\min}=2\Delta$ implies $\tau^*=0$):

\begin{proposition}
Suppose $c$ is strictly convex and $\delta_{\min}\in(\Delta,2\Delta)$. Then the untrimmed
mean dominates, at $\Delta$, any $\tau$-trimmed mean when the distribution of
unmanipulated prices has no dispersion.
\label{pronodisp1}
\end{proposition}

A similar conclusion holds, perhaps surprisingly, when variable costs are concave. In
this case the optimal manipulation again raises all manipulated prices by a common
amount. The reason is that concavity of $c$ makes it suboptimal to place any manipulated
mass strictly between zero and the upper trimming point: the manipulator can lower cost through a mean-preserving spread that
shifts this mass to the two extremes, zero and the trimming point, leaving the trimmed
mean unchanged. Once the manipulated mass is concentrated at the trimming point, the
optimal trimming parameter $\tau^*$ is the one for which the manipulator finds it optimal
to locate that point at $2\Delta$, which in turn makes the manipulated distribution
symmetric around $\Delta$. The condition that $\overline c'(q)/c'(q)$ is strictly
increasing guarantees that the manipulator's objective function has a unique interior
minimum, so that the first-order condition fully characterises the global optimum.
Economically, it requires that either fixed costs are sufficiently large or that the marginal cost of manipulation diminishes sufficiently fast as the size of distortion grows. This is a mild restriction: I verify in Appendix~\ref{appCPMM} that it is satisfied by the cost function generated by a CPMM, which is the empirically relevant case in DeFi, even when there are no fixed costs. A simple parametric family that also satisfies the condition is the class of power costs $c(q)=q^\alpha$ with $\alpha\in(0,1)$. For this family, one can compute directly that
\[
\frac{\overline c'(q)}{c'(q)} = \frac{(\alpha-1)q^\alpha - k}{\alpha\, q^{\alpha+1}},
\]
whose derivative with respect to $q$ equals $\bigl[(1-\alpha)q^\alpha + k(1+\alpha)\bigr]/(\alpha q^{\alpha+2}) > 0$ for all $q>0$ and $k\geq 0$, confirming that the condition holds throughout this class.

\begin{remark}
\label{remcompstat}
Two comparative statics of the optimal trimming $\tau^*$ in Proposition~\ref{prod2P} are worth noting. First, $\tau^*$ rises with the fixed cost of manipulation, approaching the median as fixed costs grow large relative to variable costs. Second, when the variable cost arises from trading against a constant-product market maker, $\tau^*$ falls with the liquidity of the pool: prices from deeper pools, where manipulation is more costly at the margin, are aggregated with less trimming, hence closer to the mean, whereas prices from shallower pools call for more trimming and, in the limit, the median. The first claim is direct while the second is formalised in Appendix~\ref{appCPMM}.
\end{remark}

Although the focus on the case of no dispersion is driven by tractability, the prevalence
of extensive arbitrage facilitated by bots, as discussed in \cite{daia20flash}, implies
that prices in practice often cluster tightly around their central value, making this a
reasonable assumption in the context of DeFi. Moreover, by the same continuity argument as for Proposition~\ref{profixvar}, the ranking established in Proposition~\ref{prod2P} is robust to small amounts of dispersion: the $\tau^*$-trimmed mean remains approximately optimal when the unmanipulated prices are concentrated within an $\epsilon$ of $P$, with an approximation error that vanishes as $\epsilon$ goes to zero.


\section{Benchmarks Based on Data from a Heterogeneous Market}
\label{secHet}

In this section, I consider the case in which the price distribution is generated by
different subpopulations that may differ in their manipulation costs. Such heterogeneity
may arise, for instance, when the benchmark is computed using prices observed at different
points in time and manipulation costs vary over time. A natural example is time-varying
liquidity: when liquidity is lower, a given price distortion can be achieved with smaller
orders and therefore lower variable costs. Another example, in the context of blockchains,
is time variation in the costs of transaction inclusion and ordering like gas fees and the
price in MEV auctions; see Footnote~\ref{foogas}.

I also allow for limited heterogeneity across subpopulations in the distribution of unmanipulated prices. Formally, I assume a countable set of
subpopulations $I$, where subpopulation $i\in I$ has mass $\mu_i$\footnote{$\mu_i$ represents the frequency (total mass) of prices drawn from subpopulation $i$. These masses need not sum to one; they appear in the analysis only through the normalized per-price weights in Proposition~\ref{proweighted}, which are constructed to integrate to one across all prices. Interpreting $\mu_i$ as the fraction of prices from subpopulation $i$ (so that $\sum_{i\in I}\mu_i=1$) is the leading case and involves no loss of generality.} and is described by a
probability distribution $F_i$. As before, I consider two cases: the case in which each $F_i$ has a density $f_i$ that is single-peaked and symmetric around
a common point $P$, and the case with no dispersion in which each $F_i$ is degenerate at a common point $P$.

The manipulator's strategy is described by a family, one per subpopulation, of conditional
distribution functions $\{G^c_i(\cdot|p)\}_{i\in I}$, that describe how each price in each
subpopulation is shifted by the manipulator. The manipulator's strategy gives rise to a
final distribution of prices for each subpopulation $i$:
\[G_i( p)\equiv \int G^c_i(p| \tilde p)\,  \mu_i\,  F_i\left(  d \tilde p \right).\]

The cost of a manipulating strategy $\{G^c_i(\cdot|\cdot)\}_{i\in I}$ is equal to:
\[\sum_{i\in I} \int \int_{{ p\not =\tilde p}} \left(  c_i( p-\tilde p) +k_i \right)
\,G^c_i(d  p| \tilde p)\,  \mu_i\,  F_i\left( d \tilde p\right),\]
where the $c_i$'s are continuous, strictly increasing functions (except when $c_i(\cdot)=0$),
equal to zero at zero and symmetric around zero, and $k_i\geq 0$. Thus, the $c_i$'s are
the {\it variable costs} of manipulation and the $k_i$'s the {\it fixed costs} of
manipulation.

A benchmark in this section is again a unidimensional statistic that summarizes the
distribution of prices. The novelty in this section is that the benchmark is allowed to
depend on the whole vector of distributions of final prices of all subpopulations,
$(G_i(p))_{i\in I}$, rather than on the population distribution $G$. One such example is a
weighted mean of the medians of each of the subpopulations. More generally, a weighted benchmark can be defined using some positive weights $\{w_i\}_{i\in I}$. These weights aggregate the subpopulation distributions into a single pooled probability distribution
\[
\bar G^{(w)}\;\equiv\;\frac{1}{\sum_{j\in I}w_j\mu_j}\sum_{i\in I}w_i\mu_i\,G_i,
\]
and standard statistics of $\bar G^{(w)}$ (e.g. mean, median, or $\tau$-trimmed mean) yield well-defined benchmarks. 

I call
{\it symmetric${}^*$ benchmarks} those that are equal to the point of symmetry whenever
all subpopulations are symmetric around the same point. This is clearly a larger family of
benchmarks than the symmetric benchmarks. For instance, the mean of medians is not symmetric but it is symmetric${}^*$. Still, for the distributions of the
subpopulations of unmanipulated prices any symmetric${}^*$ benchmark is equal to $P$.
Thus, the bias created by manipulation can be measured as deviations from $P$.

\begin{definition}
A benchmark is optimal$^*$ at $\Delta$ if it is the symmetric$^*$ benchmark that
maximizes the minimum cost of inducing bias $\Delta$. A benchmark is
optimal$^*$ if it is optimal$^*$ for every $\Delta$.
\end{definition}

The first main result of this section characterizes the optimal benchmark in two polar cases:
when manipulation involves no fixed costs and variable costs are convex, and when it involves no variable costs. 

\begin{proposition}
The mean of $\bar G^{(w)}$ with weights $w_i=c_i'(\Delta)$ is an optimal${}^*$ benchmark in the absence of fixed costs, i.e.\ $k_i=0$ for all $i\in I$, provided variable costs are convex.\footnote{As is immediate, the proposed benchmark is independent of $\Delta$ whenever the variable cost functions satisfy $c_i(x)=\beta_i c(x)$ for some common function $c$.} The median of $\bar G^{(w)}$ with weights $w_i=k_i$ is an optimal${}^*$ benchmark in the absence of variable costs, i.e.\ $c_i(\cdot)=0$ for all $i\in I$.
\label{proweighted}
\end{proposition}

When manipulation entails only convex variable costs, the marginal cost of distorting prices
differs across subpopulations, and the manipulator optimally reallocates manipulation
toward those for which distortions are cheapest at the margin. An optimal benchmark
therefore assigns lower weight to prices that are easier to manipulate at the margin and
higher weight to prices that are more costly to distort, which leads to a weighted mean
with weights proportional to marginal manipulation costs. In contrast, when manipulation
entails only fixed costs, the relevant trade-off is not how much to manipulate a given
price but whether to manipulate it at all. In this case, robustness requires limiting the
influence of prices that are cheap to include in a manipulation. Accordingly, the optimal
weights are proportional to the fixed costs of manipulation.

Proposition~\ref{proweighted} characterises the optimal${}^*$ benchmark in the two polar cases. When both fixed and variable costs are present, a further characterisation is available in the \emph{proportional sub-case} in which the ratio of fixed to variable cost is the same across all subpopulations.

\begin{proposition}
\label{prophybrid}
Suppose the distribution of unmanipulated prices has no dispersion, variable and fixed costs are proportional, i.e. $c_i=\alpha_i c$ and $k_i=\alpha_i\, \kappa$ for $c$ strictly increasing, symmetric with $c(0)=0$ and $\kappa\geq 0$. Let $\overline{c}_\kappa(q)\equiv(c(q)+\kappa)/q$ and $\delta_{\kappa,\min}$ its minimum. Then the optimal${}^*$ benchmark is the $\tau^*$-trimmed mean of $\bar G^{(w)}$ with $w_i=\alpha_i$,
where
\[
\tau^*\;=\;\frac{1}{2}\!\left(1-\frac{c'(2\Delta)}{\overline{c}_\kappa(2\Delta)}\right),
\]
if either:
\begin{itemize}
\item $c$ is strictly convex, $\kappa>0$, and $\delta_{\kappa,\min}>2\Delta$;
\item $c$ is strictly concave, $\kappa\geq 0$, and $\overline c_\kappa'(q)/c'(q)$ is strictly increasing.
\end{itemize}
\end{proposition}

The key to the proof is that the proportionality condition $c_i=\alpha_i\,c$ and $k_i=\kappa\alpha_i$ makes the manipulation cost depend on the strategy only through $\bar G^{(\alpha)}$: $A\int_{\tilde p\neq P}[c(\tilde p-P)+\kappa]\,d\bar G^{(\alpha)}(\tilde p)$. Since the proposed benchmark is the $\tau^*$-trimmed mean of $\bar G^{(\alpha)}$, there exists an optimal manipulative strategy that just shifts half the mass an amount $2\Delta$ by application of the results of the homogeneous case of Proposition~\ref{prod2P}. This optimal symmetric manipulation can be implemented by shifting half the mass of each subpopulation $2\Delta$, which is a symmetric${}^*$ strategy, and thus implies the result by a generalisation of Proposition \ref{prosym}.

When the proportionality condition fails, the analysis becomes substantially harder. The manipulation cost decomposes into two separately weighted aggregations of the per-subpopulation strategy: one weighted by marginal costs and one weighted by fixed costs. When these weighting schemes disagree, a single benchmark cannot simultaneously induce optimality under both, and no choice of $\tau$ resolves the tension. Characterising the optimal${}^*$ benchmark in the general intermediate case is left for future work.

\section{Conclusion}
\label{secConclusion}

This paper links the optimal benchmark formula to the cost of its manipulation. When only variable costs matter, the mean maximises resistance to manipulation; when only fixed costs matter, the median does so; and when both are present, a trimmed mean interpolates between the two, with the degree of trimming increasing in the relative importance of fixed costs. This correspondence provides an explanation for an empirical pattern: benchmarks in traditional finance, where manipulation requires sustained trading, tend to be means, while those in DeFi, where flash loans can eliminate variable costs but securing transaction ordering is costly, tend to be medians.

A few simplifying assumptions bound the scope of the analysis. My characterisations of optimal benchmarks rely on the assumption that the distribution of unmanipulated prices is symmetric or has no dispersion. This allows for clear cut results in a tractable setting. Extending the analysis to either asymmetric distribution of prices or dispersed price distributions is challenging, as the manipulator's problem becomes a non-separable functional optimisation problem. In the heterogeneous setting, Proposition~\ref{prophybrid} requires that fixed and variable costs are proportional across subpopulations: when fixed and variable costs vary independently the optimal benchmark remains uncharacterised, because the manipulation cost decomposes into two separately weighted aggregations that no single benchmark can reconcile.

Several directions for future research stand out. One speaks directly to the limitations just discussed. Restricting attention to linear variable costs, at some cost in empirical grounding, would recast the benchmark design problem as a maximin variant of the optimal partial transport problem. This would not by itself restore separability, but it would bring the problem within reach of duality-based methods, which may permit progress on asymmetric price distributions.

Two further directions broaden the framework rather than relax its assumptions. The first concerns the sampling of prices: this paper takes the set of sampled prices as given and optimises only the aggregation formula, whereas jointly optimising the sampling protocol and the formula (for example, choosing which transaction times or which liquidity pools to include) could yield additional gains in manipulation resistance, especially in decentralised exchanges where liquidity varies sharply across time and venue. The second concerns statistical efficiency: this paper isolates the manipulation-resistance margin and takes the point of symmetry as the estimation target, setting aside the precision of the benchmark when prices are noisy but unmanipulated. Because the optimal benchmarks belong to the L-estimator family, a natural next step is to choose the trimming fraction and weights so as to balance manipulation resistance against asymptotic efficiency. This would trace out the frontier between efficiency and robustness studied in robust statistics, see \cite{huber}, connecting the present analysis to the trade-off between efficiency and manipulation emphasised by \cite{duffie2021robust}.

\appendix

\section{Proofs}
\label{appProofs}

\subsection{Proof of Proposition \ref{prosym}}

\begin{proof}
Consider a benchmark for which there exists such cost-minimizing manipulation strategy. Since the strategy induces a symmetric distribution of manipulated prices, it should implement $\Delta$ not only for this benchmark, but for any symmetric benchmark. Hence, the minimum cost required to implement $\Delta$ under any symmetric benchmark is weakly bounded above by the cost of this strategy, which implies the result.
\end{proof}

\subsection{Proof of Proposition \ref{profixvar}}

\begin{proof}
Begin with the case $k=0$ and fix an arbitrary bias $\Delta>0$ (recall that the analysis for $\Delta<0$ is symmetric). In the absence of fixed costs and given the convexity of $c$, the cost-minimizing way to raise the mean in $\Delta$ is to increase all prices by the same amount $\Delta$. This uniform shift raises the entire distribution by $\Delta$ and therefore preserves symmetry. The first result then follows directly from Proposition~\ref{prosym}.

Next consider the case $c(\cdot)=0$ and an arbitrary bias $\Delta>0$. Let $P$ denote the point of symmetry of the distribution of unmanipulated prices.
For any $y\in \real$, define,
\[
\tilde g(y)\equiv \frac{f(y)+f(y-2\Delta)}{2}, \]
and for any $y<P+\Delta$,
\[\alpha(y)\equiv \frac{\tilde g(y)}{f(y)},
\qquad
A(y)\equiv
\frac{\int_{P+\Delta}^\infty \frac{f(z-2\Delta)-f(z)}{2}\,dz}{1-\alpha(y)}.
\]
The symmetry of $f$ around $P$ and its strict single-peakedness imply that $f(y)>f(y-2\Delta)$ for all $y<P+\Delta$: since $y-2\Delta$ is further from $P$ than $y$ whenever $y<P+\Delta$ and $\Delta>0$, single-peakedness gives $f(y-2\Delta)<f(y)$, so $\alpha(y)\in(0,1)$.

The manipulation strategy is defined as follows. For any $y\geq P+\Delta$, the conditional distribution $G^c(\cdot\mid y)$ is degenerate at $y$, i.e. there is no manipulation. For any $y<P+\Delta$, $G^c(\cdot\mid y)$ has a mass point $\alpha(y)$ at $y$ and density:
\[
\frac{\tilde g(x)-f(x)}{A(y)}
\quad\text{for  } x\in [P+\Delta,\infty).
\]

The total mass shifted by this strategy is:
\begin{eqnarray*}
\int_{-\infty}^{P+\Delta} \int_{P+\Delta}^{\infty} \frac{\tilde g(x)-f(x)}{A(y)} \,dx\, f(y)\,dy
&=& \int_{-\infty}^{P+\Delta} (1-\alpha(y))
\frac{\int_{P+\Delta}^{\infty} \left(\tilde g(x)-f(x)\right)\,dx}
{\int_{P+\Delta}^\infty \frac{f(z-2\Delta)-f(z)}{2}\,dz}
f(y)\,dy \\
&=& \int_{-\infty}^{P+\Delta} (1-\alpha(y)) f(y)\,dy \\
&=& \int_{-\infty}^{P+\Delta} \frac{f(y)-f(y-2\Delta)}{2}\,dy \\
&=& \frac{F(P+\Delta)-F(P-\Delta)}{2}
= F(P+\Delta)-\frac{1}{2},
\end{eqnarray*}
where the equalities use, respectively, the definition of $A(y)$, the definition of $\tilde g$, the definitions of $\alpha(y)$ and $\tilde g(y)$, and finally the symmetry of $f$ around $P$.

To conclude, it remains to verify that the induced distribution of manipulated prices is symmetric around $P+\Delta$. First, we show that this density is $\tilde g$. For $x<P+\Delta$, the density equals $\alpha(x)f(x)=\tilde g(x)$. For $x\geq P+\Delta$, it is given by
\begin{eqnarray*}
f(x)+ \int_{-\infty}^{P+\Delta} \frac{\tilde g(x)-f(x)}{A(y)} f(y)\,dy
&=& f(x)+\bigl(\tilde g(x)-f(x)\bigr)
\frac{\int_{-\infty}^{P+\Delta} (1-\alpha(y))f(y)\,dy}
{\int_{P+\Delta}^\infty \frac{f(z-2\Delta)-f(z)}{2}\,dz} \\
&=& f(x)+\bigl(\tilde g(x)-f(x)\bigr)
\frac{\int_{-\infty}^{P+\Delta} \left( f(y) - f(y-2\Delta) \right)  \,dy}
{\int_{P+\Delta}^\infty \left(  f(z-2\Delta)-f(z) \right)  \,dz} \\
&=& \tilde g(x),
\end{eqnarray*}
where the derivation again uses the definitions of $A(y)$, $\alpha(y)$, and $\tilde g$, together with the identity:
\begin{equation*}
\frac{\int_{-\infty}^{P+\Delta} \left( f(y) - f(y-2\Delta) \right)  \,dy}
{\int_{P+\Delta}^\infty \left(  f(z-2\Delta)-f(z) \right)  \,dz} =
\frac{\int_{\Delta}^{\infty} \left( f(P+2\Delta-t) - f(P-t) \right)  \,dt}
{\int_{\Delta}^\infty \left(  f(P+s-2\Delta)-f(P+s) \right)  \,ds} = 1,
\end{equation*}
where the first step follows from the change of variables $t=P+2\Delta-y$ and $s=z-P$, and the second from the symmetry of $f(P+x)=f(P-x)$ by symmetry around $P$.

Finally, we show that $\tilde g$ is symmetric around $P+\Delta$:
\begin{eqnarray*}
\tilde g(P+\Delta-x)&=&\frac{f(P+\Delta-x)+f(P-\Delta-x)}{2}\\
&=&\frac{f(P-\Delta+x)+f(P+\Delta+x)}{2}\\
&=&\tilde g(P+\Delta+x),
\end{eqnarray*}
where the first step and last step follows from the definition of $\tilde g$, the second from the symmetry of $f$ around $P$.

This completes the proof.
\end{proof}

\subsection{Proof of Proposition \ref{prodP}}

\begin{proof}
The strict convexity of the variable cost function implies that any cost-minimizing manipulation of the mean must increase all manipulated prices by the same amount. Let $\delta$ denote this increase and let $\mu$ denote the measure of manipulated prices. The optimal manipulation then solves
\[
\left\{
\begin{array}{cl}
\min\limits_{\delta\in\mathbb{R}_+,\;\mu\in[0,1]} & \mu\bigl(k+c(\delta)\bigr) \\
\text{s.t.} & \delta\,\mu = \Delta .
\end{array}
\right.
\]
Using the constraint to eliminate $\mu$, this problem can be written as
\[
\min_{\delta\in[\Delta,\infty)} \Delta\,\overline c(\delta).
\]

If $\delta_{\min}\leq \Delta$, the solution is $\delta=\Delta$ and $\mu=1$. If instead $\delta_{\min}=2\Delta$, the solution is $\delta=\delta_{\min}$ and $\mu=1/2$. In both cases, the corresponding manipulation strategy can induce a final distribution of prices that is symmetric around $\Delta$. In the former case, symmetry follows because the entire distribution is shifted by the same amount and the original distribution is symmetric. In the latter case, symmetry follows because it suffices to shift exactly the upper half of the price distribution. Proposition~\ref{prosym} therefore applies, which establishes the result.
\end{proof}

\subsection{Proof of Proposition \ref{prod2P}}

\begin{proof}
As a preliminary step we argue that $\tau^*\in (0,1/2)$. That $\tau^*<1/2$ follows immediately from $c'(2\Delta)>0$. That $\tau^*>0$, i.e.\
$\overline c(2\Delta)>c'(2\Delta)$, is verified separately in each case below. We also argue in each of the two cases that any cost-minimizing
manipulation takes the form of increasing a mass $\mu$ of prices by a common amount
$\delta$. Raising the benchmark from $P$ to $P+\Delta$ then requires
$\delta(\mu-\tau^*)/(1-2\tau^*)=\Delta$ with $\mu\in(\tau^*,1-\tau^*]$. Eliminating
$\mu=\tau^*+(1-2\tau^*)\Delta/\delta$, the feasibility condition $\mu\leq 1-\tau^*$ reduces to
$\delta\geq\Delta$, and the manipulation problem becomes
\begin{equation}
\label{eqMIN}
\min_{\delta\geq\Delta}
\left\{
\Delta(1-2\tau^*)\,\overline c(\delta)+\tau^*\bigl(k+c(\delta)\bigr)
\right\},
\end{equation}
whose derivative with respect to $\delta$ is
\begin{equation}
\label{eqFOC}
c'(\delta)\,\Delta(1-2\tau^*)
\left(\frac{\overline c'(\delta)}{c'(\delta)}
+\frac{\tau^*}{\Delta(1-2\tau^*)}\right).
\end{equation}
The definition of $\tau^*$ together with $\overline
c'(\delta)=(c'(\delta)-\overline c(\delta))/\delta$ means that $\delta=2\Delta$ makes this derivative equal to zero. We show in the two cases below, that $\delta=2\Delta$ is not only a local but a global minimizer. So the optimally manipulated distribution places mass $1/2$ at $P$ and
mass $1/2$ at $P+2\Delta$. This distribution is symmetric around $P+\Delta$, so the optimality of the $\tau^*$-trimmed mean follows from Proposition~\ref{prosym}.

\medskip
\noindent\textit{Case 1: $c$ strictly convex.} Strict convexity implies that a
cost-minimizing manipulation sets a common $\delta$ across all manipulated prices.
Since $\delta_{\min}>2\Delta$, $\overline c(\delta)>c'(\delta)$ for all
$\delta\in(0,\delta_{\min})$, so in particular $\overline c(2\Delta)>c'(2\Delta)$ and
$\tau^*>0$. For the global minimum: strict convexity implies $\overline c(\delta)-c'(\delta)$ is
positive and strictly decreasing on $(0,\delta_{\min})$, so \eqref{eqFOC} is strictly
negative for $\delta\in(0,2\Delta)$ and strictly positive for
$\delta\in(2\Delta,\delta_{\min})$; for $\delta\geq\delta_{\min}$, $\overline
c(\delta)\leq c'(\delta)$ so \eqref{eqFOC} is strictly positive. Hence the objective
in \eqref{eqMIN} is strictly decreasing then strictly increasing, and $\delta=2\Delta$
is its unique global minimiser.

\medskip
\noindent\textit{Case 2: $c$ strictly concave.} For any given manipulation, let $p^*$
be the $(1-\tau^*)$-quantile of the final price distribution. It is never optimal to place
manipulated mass above $p^*$, as doing so is costly without raising the $\tau^*$-trimmed
mean. For the untrimmed mass in $(0,p^*)$, concentrating it onto $\{0,p^*\}$ while
preserving its conditional mean is a mean-preserving spread, which strictly decreases
expected variable cost when $c$ is concave and also reduces expected fixed cost, while
leaving the $\tau^*$-trimmed mean unchanged. Hence a common $\delta$ is again optimal and
the problem reduces to \eqref{eqMIN}. Strict concavity of $c$ with $c(0)=0$ gives
$c(2\Delta)>2\Delta c'(2\Delta)$, and therefore $\overline{c}(2\Delta)=(k+c(2\Delta))/(2\Delta)>c'(2\Delta)$
for all $k\geq 0$, so $\tau^*>0$. For the global
minimum: since $c'(\delta)>0$ and $\Delta(1-2\tau^*)>0$, the sign of \eqref{eqFOC} is
determined by the bracket, which is strictly increasing in $\delta$ by assumption, equals
zero at $\delta=2\Delta$, and therefore is strictly negative for $\delta<2\Delta$ and
strictly positive for $\delta>2\Delta$. Hence the objective in \eqref{eqMIN} is strictly
decreasing then strictly increasing, and $\delta=2\Delta$ is its unique global minimiser.
\end{proof}

\subsection{Proof of Proposition \ref{pronodisp1}}

\begin{proof}
The argument follows the proof of Proposition~\ref{prod2P}. Strict convexity of $c$
implies that, for a $\tau$-trimmed mean, the cost-minimizing manipulation consists of
increasing a measure of prices
\[
\mu=\tau+(1-2\tau)\frac{\Delta}{\delta^*}
\]
by a common amount $\delta^*\in[\Delta,\delta_{\min}]$, where $\delta^*$ solves
\eqref{eqMIN}. The value of the objective in \eqref{eqMIN}, evaluated at $\delta^*$,
equals the minimum cost of manipulating the $\tau$-trimmed mean.

It is immediate from \eqref{eqMIN} that this cost is strictly decreasing in $\tau$ in the
case of the corner solution $\delta^*=\Delta$. If instead $\delta^*>\Delta$, then
necessarily $\delta^*\leq\delta_{\min}$: for $\delta>\delta_{\min}$, $\overline c'(\delta)>0$
so the bracket in \eqref{eqFOC} is strictly positive and the objective in \eqref{eqMIN} is
strictly increasing, ruling out an interior minimum above $\delta_{\min}$. In this case the
envelope theorem implies that the derivative of the minimum manipulation cost with
respect to $\tau$ is
\[
-\frac{k+c(\delta^*)}{\delta^*}\,(2\Delta-\delta^*)\;<\;0,
\]
where the inequality follows from $\delta^*\leq\delta_{\min}<2\Delta$. Hence, in all
cases, increasing $\tau$ strictly reduces the minimum cost of manipulation. It follows
that the untrimmed mean ($\tau=0$) maximizes the cost of manipulation at $\Delta$ and
therefore dominates any trimmed mean.
\end{proof}

\subsection{Proof of Proposition \ref{proweighted}}

\begin{proof}
Consider, first, the case of strictly convex variable costs and no fixed costs.
The argument characterizes the optimal manipulation strategy for the proposed benchmark. Since
manipulation costs are strictly convex, it is optimal to increase all prices within the
same subpopulation by the same amount. Let $\delta_i$ denote the increase applied to
prices in subpopulation $i$. The optimal vector $\{\delta_i\}_{i\in I}$ solves
\begin{eqnarray*}
\min_{\{\delta_i\}_{i\in I}}
& \sum_{i\in I} c_i(\delta_i)\,\mu_i \\
\text{s.t.}
& \sum_{i\in I}
\frac{c_i'(\Delta)}{\sum_{j\in I} c_j'(\Delta)\mu_j}\,
\mu_i\,\delta_i
= \Delta.
\end{eqnarray*}

This problem has a convex objective function with a linear constraint, and thus it is sufficient to check the first-order conditions, which are readily verified to hold at $\delta_i=\Delta$ for all $i\in I$. Hence, the optimal manipulation
consists of shifting all prices by the same amount $\Delta$. Because each subpopulation of
unmanipulated prices is symmetric around $P$, the corresponding distributions of
manipulated prices are symmetric around $P+\Delta$. Consequently, this manipulation
implements the value $\Delta$ for any symmetric${}^*$ benchmark, which implies that the
proposed benchmark is optimal${}^*$ at $\Delta$.

Consider next the case of no variable costs. Let
\[
F(P+\Delta)\equiv\sum_{i\in I}\frac{k_i}{\sum_{j\in I}k_j\,\mu_j}\,\mu_i\,F_i(P+\Delta)
\]
denote the weighted fraction of unmanipulated prices lying below $P+\Delta$ under the
proposed benchmark. Since each $F_i$ is symmetric around $P$, $F_i(P+\Delta)>\tfrac{1}{2}$
for all $\Delta>0$, so $F(P+\Delta)>\tfrac{1}{2}$. The least costly way to raise the proposed benchmark to $P+\Delta$ requires shifting a
weighted mass of prices from strictly below $P+\Delta$ to weakly above $P+\Delta$ until the
cumulative weighted mass weakly below $P+\Delta$ is exactly $1/2$. That is, the mass of prices $m_i$ in
each subpopulation $i$ that is shifted from strictly below $P+\Delta$ to weakly above
$P+\Delta$ must satisfy:
\[
\sum_{i\in I} \frac{k_i}{\sum_{j\in I} k_j\,\mu_j}\, m_i
= F(P+\Delta)-\frac{1}{2}.
\]
Hence, the cost of any such manipulation strategy is equal to:
\[
\sum_{i\in I} k_i\, m_i= \left(F(P+\Delta)-\frac{1}{2}\right) \sum_{i\in I} k_i\,\mu_i,
\]
and therefore does not depend on how the total weighted mass $F(P+\Delta)-1/2$ is allocated
across subpopulations.

One feasible manipulation satisfying this condition is obtained by setting:
\[
m_i=\mu_i\left(F_i(P+\Delta)-\frac{1}{2}\right)
\]
and applying, within each subpopulation, an adaptation of the optimal manipulation
strategy used in the proof of Proposition~\ref{profixvar}. In particular, for any $i\in I$
and $y\in\real$, define:
\[
\tilde g_i(y)\equiv \frac{f_i(y)+f_i(y-2\Delta)}{2},
\]
and for any $y<P+\Delta$,
\[
\alpha_i(y)\equiv \frac{\tilde g_i(y)}{f_i(y)},
\qquad
A_i(y)\equiv
\frac{\int_{P+\Delta}^\infty \frac{f_i(z-2\Delta)-f_i(z)}{2}\,dz}{1-\alpha_i(y)}.
\]

The manipulation strategy for subpopulation $i$ is defined as follows. For any
$y\geq P+ \Delta$, the conditional distribution $G_i^c(\cdot\mid y)$ is degenerate at $y$,
i.e.\ prices are left unchanged. For any $y<P+\Delta$, $G_i^c(\cdot\mid y)$ assigns
probability $\alpha_i(y)$ to $y$ and has density
\[
\frac{\tilde g_i(x)-f_i(x)}{A_i(y)}
\quad\text{for } x\in[P+\Delta,\infty).
\]

The same arguments as in Proposition~\ref{profixvar} imply that the total mass shifted in
subpopulation $i$ from strictly below $P+\Delta$ to weakly above $P+\Delta$ is
$\mu_i\left(F_i(P+\Delta)-\frac{1}{2}\right)$, and that the induced distribution of
manipulated prices within each subpopulation is symmetric around $P+\Delta$. Consequently,
this manipulation implements $\Delta$ for any symmetric${}^*$ benchmark, which
implies that the proposed benchmark is optimal${}^*$ by an adaptation of the arguments in Proposition \ref{prosym}.
\end{proof}

\subsection{Proof of Proposition \ref{prophybrid}}

\begin{proof} Since the distribution of unmanipulated prices is degenerate at $P$, any manipulation strategy has total cost:
\begin{align*}
C &=  \sum_{i\in I}\mu_i\!\int_{\tilde p\neq P}\!\bigl[\alpha_i\,c(\tilde p-P)+\alpha_i \kappa\bigr]\,dG_i(\tilde p)\\
&= \left( \sum_j\alpha_j\mu_j \right) \!\int_{\tilde p\neq P}\!\bigl[c(\tilde p-P)+\kappa\bigr]\,d\bar G^{(\alpha)}(\tilde p),
\end{align*}
The cost depends on the strategy only through $\bar G^{(\alpha)}$. The proposed benchmark is the $\tau$-trimmed mean of $\bar G^{(\alpha)}$ by definition, so the benchmark value also depends on the strategy only through $\bar G^{(\alpha)}$.

The problem therefore reduces to: choose $\tau$ to maximise the minimum cost of making the $\tau$-trimmed mean of $\bar G^{(\alpha)}$ equal to $P+\Delta$, where the cost per unit mass of $\bar G^{(\alpha)}$ displaced from $P$ is $\left( \sum_j\alpha_j\mu_j \right) \left( c(\tilde p-P)+\kappa\right)$. Every symmetric distribution of $\bar G^{(\alpha)}$ around $P+\Delta$ is achievable by setting each $G_i^c(\cdot|P)$ equal to that distribution, so this is identical to the homogeneous problem of Proposition~\ref{prod2P} with variable cost $c$ and fixed cost $\kappa$. Under the stated regularity condition, that proposition yields $\tau^*=\frac{1}{2}(1-c'(2\Delta)/\overline{c}_\kappa(2\Delta))$ as desired.
\end{proof}

\section{The Manipulation Cost of the CPMM}
\label{appCPMM}

A \emph{constant-product market maker} (CPMM) is an automated liquidity pool that holds reserves of two assets, say $X$ and $Y$. Any trader may swap one asset for the other: given a deposit of $\Delta y$ units of $Y$, the pool releases however much $X$ is needed to keep the product of the two reserve quantities constant. A proportional fee rate $f\in(0,1)$ is applied to every trade, so that only the \emph{effective input} $(1-f)\Delta y$ counts toward the constant-product invariant that determines how much $X$ is released; the remaining fraction $f\Delta y$ stays in the pool as revenue for the liquidity providers. The fee makes the round-trip cost of a manipulation that pushes the price upward and then unwinds it strictly positive, even when the price is fully restored to its original level. This appendix computes this round-trip cost as a function of the price distortion $q$ and verifies that it is strictly concave and satisfies the regularity condition required by Proposition~\ref{prod2P}.

Consider a CPMM pool, where a constant-product market maker maintains reserves $(x,y)$ of assets $X$ and $Y$ satisfying the invariant $xy=\kappa$ for some $\kappa>0$. The marginal price of $X$ in terms of $Y$ is given by the ratio of reserves $y/x$, which follows directly from differentiating the invariant. With initial reserves $(x_0,y_0)$ and fee rate $f\in(0,1)$, let $P_0 \equiv y_0/x_0$ denote the initial marginal price. A manipulator submits a {\it push} order that buys $X$ to raise the marginal price to $P_0+q$ where $q>0$, followed by an {\it unwind} order that sells $X$ to restore the price to $P_0$.\footnote{One
could consider two alternative unwinding rules for manipulation. The first is that the
manipulator sells back exactly the quantity of $X$ received from the push order. This
leaves a residual price $P_2=P_0(1+v)/[1+(1-f)^2v]>P_0$, so the manipulation is only
partially reversed. The second is that the unwinding manipulation minimises the cost of
the manipulation valued at $P_0$. Again, this yields a final price $P_2>P_0$ and thus
there is no price restoration. Either alternative delivers similar outcomes but requires a slightly more convoluted analysis.}

\paragraph{Push.} After the submission of an order that sends $\Delta^y_{in}$ units of $Y$
to the AMM, the CPMM algorithm updates the reserves to:
\[
(x_1,y_1)\equiv\left(\frac{x_0}{1+(1-f)\delta_y},\,y_0(1+\delta_y)\right),
\qquad \delta_y\equiv\frac{\Delta^y_{in}}{y_0}.
\]
The amount of $X$ obtained is $\Delta^x_{\mathrm{out}}=\frac{(1-f)\delta_y}{1+(1-f)\delta_y}x_0$
and the new marginal price is $P_0(1+\delta_y)(1+(1-f)\delta_y)$. Setting this equal to
$P_0+q$ requires $\delta_y=\delta(q)$, the positive root of:
\begin{equation}\label{eq:push}
(1+\delta)(1+(1-f)\delta)=1+\frac{q}{P_0}.
\end{equation}

\paragraph{Unwind.} The unwind order sends $\Delta^x_{in}$ units of $X$ to the AMM. Setting
$\delta_x\equiv\Delta^x_{in}/x_1$, the CPMM algorithm gives post-unwind reserves:
\[
(x_2,y_2)\equiv\left(x_1(1+\delta_x),\,\frac{y_1}{1+(1-f)\delta_x}\right),
\]
and a new marginal price of $P_0(1+\delta_y)(1+(1-f)\delta_y)/[(1+\delta_x)(1+(1-f)\delta_x)]$.
Setting this equal to $P_0$ gives $(1+\delta_x)(1+(1-f)\delta_x)=1+q/P_0$, so
$\delta_x=\delta(q)$ and
\[
\Delta^x_{in}=\frac{\delta(q)}{1+(1-f)\delta(q)}x_0,
\qquad
\Delta^y_{\mathrm{out}}=\frac{(1-f)\delta(q)}{1+(1-f)\delta(q)}(1+\delta(q))y_0.
\]

\paragraph{Cost.} The net inflows (to the AMM) across both trades of $Y$ is:
\[
\Delta^y_{in}-\Delta^y_{\mathrm{out}}=\frac{f\,\delta(q)}{1+(1-f)\delta(q)}y_0,
\]
and of $X$ inflow:
\[
\Delta^x_{in}-\Delta^x_{\mathrm{out}}=\frac{f\,\delta(q)}{1+(1-f)\delta(q)}x_0.
\]
Since a net inflow to the AMM is a cost to the trader, the total cost for the manipulator in units of $Y$ (evaluating the position in $X$ at price $P_0=y_0/x_0$) is:
\begin{equation}\label{eq:cost}
c(q)=\frac{2f\,\delta(q)}{1+(1-f)\delta(q)}\,y_0.
\end{equation}

\paragraph{Concavity.} Implicit differentiation of \eqref{eq:push} gives
\[
\delta'(q)=\frac{1}{P_0\bigl((2-f)+2(1-f)\delta(q)\bigr)},
\]
which is strictly decreasing in $q$, so $\delta$ is strictly concave in $q$. The map
$g(\delta)\equiv 2f\delta y_0/(1+(1-f)\delta)$ satisfies $g'>0$ and $g''<0$, so it is
strictly increasing and concave. Since $c=g\circ\delta$, the chain rule gives
\[
c''(q)=g''(\delta(q))\,\delta'(q)^2+g'(\delta(q))\,\delta''(q)<0,
\]
establishing strict concavity of $c$.

\paragraph{Monotonicity of $\overline{c}'/c'$.}
This paragraph verifies that $\overline{c}'(q)/c'(q)$ is strictly increasing in $q$.
First note that
\[
\frac{\overline{c}'(q)}{c'(q)}=- \left(\underbrace{\frac{k}{q^2 c'(q)} }_{\text{fixed term}}+ \underbrace{\frac{c(q)-qc'(q)}{q^2 c'(q)}}_{\text{variable term}} \right).
\]
Both terms are positive, the variable term because $c$ is concave with $c(0)=0$, so
$c(q)\geq qc'(q)$. It therefore suffices to show that each term is strictly decreasing in
$q$. Substitute
\[
c(q)=\dfrac{2fy_0\,\delta(q)}{1+(1-f)\delta(q)},
\qquad
c'(q)=\frac{2fy_0}{\bigl(1+(1-f)\delta(q)\bigr)^2}\,\delta'(q),
\]
and change variable to $d=\delta(q)$, which is strictly increasing in $q$. Letting
\[
\nu(d)\equiv\delta^{-1}(d)=P_0\bigl((1+d)(1+(1-f)d)-1\bigr)=P_0\,d\bigl(2-f+(1-f)d\bigr),
\]
one has $\delta(\nu(d))=d$ and
\[
\delta'(\nu(d))=\frac{1}{\nu'(d)}=\frac{1}{P_0\bigl(2-f+2(1-f)d\bigr)}.
\]

\emph{Fixed term.} Substituting,
\begin{eqnarray*}
\frac{k}{\nu(d)^2\, c'(\nu(d))}s
&=&\frac{k}{P_0^2d^2 \bigl(2-f+(1-f)d\bigr)^2 \cdot \frac{2fy_0}{(1+(1-f)d)^2}\, \frac{1}{P_0 (2-f+2(1-f)d)} }\\
&=&\frac{k}{2fy_0 P_0} \cdot \left( \frac{1+(1-f)d}{d} \right)^2 \cdot \frac{2-f+2(1-f)d}{ \bigl(2-f+(1-f)d\bigr)^2}.
\end{eqnarray*}
The first factor is constant, the second is decreasing because
$(1+(1-f)d)/d=(1/d)+(1-f)$ is, and the third is decreasing because its logarithmic
derivative equals
$2(1-f)\bigl(\tfrac{1}{2-f+2(1-f)d}-\tfrac{1}{2-f+(1-f)d}\bigr)<0$. Hence the fixed term
is strictly decreasing in $d$.

\emph{Variable term.} Substituting,
\begin{eqnarray*}
\frac{ c(\nu(d))-\nu(d) c'(\nu(d))}{\nu(d)^2 c'(\nu(d))}
&=&\frac{\frac{2fy_0d}{1+(1-f)d} - P_0d(2-f+(1-f)d) \cdot \frac{2fy_0}{(1+(1-f)d)^2} \frac{1}{P_0(2-f+2(1-f)d)}}{P_0^2 d^2(2-f+(1-f)d)^2 \cdot \frac{2fy_0}{(1+(1-f)d)^2}\frac{1}{ P_0 (2-f+2(1-f)d)} }\\
&=&\frac{(1+(1-f)d) -  \frac{2-f+(1-f)d}{2-f+2(1-f)d}}{P_0 d(2-f+(1-f)d)^2 \cdot \frac{1}{ 2-f+2(1-f)d} }\\
&=&\frac{(1-f)d +  \frac{(1-f)d}{2-f+2(1-f)d}}{P_0 d(2-f+(1-f)d)^2 \cdot \frac{1}{ 2-f+2(1-f)d} }\\
&=&\frac{(1-f) \left( 3-f+2(1-f)d  \right)}{P_0 (2-f+(1-f)d)^2  }\\
&=& \frac{1-f}{P_0\bigl(2-f+(1-f)d\bigr)} \left(  2- \frac{ 1-f  }{ 2-f+(1-f)d  } \right).
\end{eqnarray*}
The first equality uses the definitions of $c$, $\nu$, and $\delta'$. The second cancels
the common factor $\frac{2fy_0 d}{(1+(1-f)d)^2}$ and simplifies the terms in $P_0$. The
third uses $2-f+(1-f)d=\bigl(2-f+2(1-f)d\bigr)-(1-f)d$. The fourth multiplies through by
$2-f+2(1-f)d$, factors $(1-f)d$ from the numerator, and cancels $d$. The fifth rewrites
$3-f+2(1-f)d=2\bigl(2-f+(1-f)d\bigr)-(1-f)$ and factors out
$\frac{1-f}{P_0(2-f+(1-f)d)}$.

Writing $u(d)\equiv (1-f)/\bigl(2-f+(1-f)d\bigr)$, the last expression equals
$\tfrac{1}{P_0}\,u(d)\bigl(2-u(d)\bigr)$. Since $u(d)<1$ and $u'(d)<0$, and $h(u)=u(2-u)$
is strictly increasing for $u<1$, the variable term is strictly decreasing in $d$.

Since both terms are strictly decreasing in $d$, and $\delta(q)$ is strictly increasing in $q$,
$\overline{c}'(q)/c'(q)$ is strictly increasing in $q$, as required.

\paragraph{Comparative statics of the optimal trimming.}
By Proposition~\ref{prod2P}, the optimal trimming is 
\[\tau^*=\tfrac12-\Delta\,\frac{c'(2\Delta)}{c(2\Delta)+k},\] where $c(q)=y_0\,\gamma(q)$ with:
\[\gamma(q)\equiv 2f\frac{\delta(q)}{1+(1-f)\delta(q)}\] the cost per unit of liquidity. Holding the marginal price $P_0$ fixed, $\delta(\cdot)$, and hence $\gamma$, is independent of $y_0$, so $c'(q)=y_0\,\gamma'(q)$ with $\gamma'>0$, and
\[
\tau^*=\tfrac12-\frac{\Delta\,y_0\,\gamma'(2\Delta)}{y_0\,\gamma(2\Delta)+k},
\]
which is stricly decreasing in $y_0$ whenever $k>0$. Hence $\tau^*$ is non-increasing in the liquidity $y_0$, strictly so when $k>0$; as $y_0\to\infty$ it converges to $\tfrac12-\Delta\,\gamma'(2\Delta)/\gamma(2\Delta)$, the level induced by the variable cost alone, and as $y_0\to 0$ it converges to $1/2$.

\bibliography{auctions.bib}

\end{document}